\documentclass[3p,12pt]{elsarticle}

\usepackage{amssymb,latexsym,amsmath,color,subfigure,amsthm,url}

\journal{}

\newcommand{\eps}{\varepsilon}

\newcommand{\set}[1]{\left\{#1\right\}}

\newcommand{\p}{\partial}
\newcommand{\mc}{\mathbf{c}}
\newcommand{\mn}{\mathbf{n}}
\newcommand{\mt}{\mathbf{t}}
\newcommand{\mx}{\mathbf{x}}
\newcommand{\my}{\mathbf{y}}
\newcommand{\mz}{\mathbf{z}}

\newcommand{\mU}{\mathbf{U}}
\newcommand{\mV}{\mathbf{V}}
\newcommand{\mW}{\mathbf{W}}

\newcommand{\E}{\mathbb{E}}
\newcommand{\W}{\mathbb{W}}

\newcommand{\vt}{\boldsymbol{\theta}}
\newcommand{\vv}{\boldsymbol{\vartheta}}

\newtheorem{thm}{Theorem}[section]
\newtheorem{cor}[thm]{Corollary}
\newtheorem{lem}[thm]{Lemma}

\begin{document}

\begin{frontmatter}



\title{Analysis of multi-frequency subspace migration weighted by natural logarithmic function for fast imaging of two-dimensional thin, arc-like electromagnetic inhomogeneities}


\author{Young-Deuk Joh}
\ead{mathea421@kookmin.ac.kr}
\author{Won-Kwang Park\corref{Parkwk}}
\ead{parkwk@kookmin.ac.kr}
\address{Department of Mathematics, Kookmin University, Seoul, 136-702, Korea.}
\cortext[Parkwk]{Tel: +82 2 910 5748; fax: +82 2 910 4739.}

\begin{abstract}
The present study seeks to investigate mathematical structures of a multi-frequency subspace migration weighted by the natural logarithmic function for imaging of thin electromagnetic inhomogeneities from measured far-field pattern. To this end, we designed the algorithm and calculated the indefinite integration of square of Bessel function of order zero of the first kind multiplied by the natural logarithmic function. This is needed for mathematical analysis of improved algorithm to demonstrate the reason why proposed multi-frequency subspace migration contributes to yielding better imaging performance, compared to previously suggested subspace migration algorithm. This analysis is based on the fact that the singular vectors of the collected Multi-Static Response (MSR) matrix whose elements are the measured far-field pattern can be represented as an asymptotic expansion formula in the presence of such inhomogeneities. To support the main research results, several numerical experiments with noisy data are illustrated.
\end{abstract}


\end{frontmatter}





\section{Introduction}\label{Sec1}
The inverse scattering problem, non-destructive evaluation, is one of the intriguing research topics since it is closely related to human
life. It is because it applies not only to physics, engineering, or image medical science but also to identifying the cracks of the structures such as concrete walls, machines, or buildings. Therefore, it has been considerably investigated by many researchers to suggest the algorithm regarding this problem or to experiment and analyze previously suggested algorithms. Related works can be found in \cite{AWSH,ABF,A,CI,CIN,FKM,GKL,I,IC,ICK,IZSK,KWYS,KYSWC,SKAW,TLRD} and references therein.

However, the inverse scattering problem is such a difficult problem that not many methods have been studied other than the reconstruction method based on the iterative method such as Newton-type method, refer to \cite{ADIM,BHL,BHR,DL,PL4,S,VXB}. Regarding the algorithms of using a Newton-type method, in case the initial shape is quite different from the unknown target, the reconstruction of material leads to failure with the non-convergence or yielding faulty shapes. Even though the reconstruction ends up with a successful result, it could take a great deal of time. Therefore, several non-iterative algorithms have been suggested because they can reconstruct the shape that is quite similar to the target, and thus it can be used as a good initial guess, which also takes a short time and efficient in iterative methods (see \cite{AIL,AKKLV,AKLP,CA,CZ,CK,G,PL1,PL2,PL3,ZC} and references therein).

Among them, the non-iterative reconstruction algorithm such as Kirchhoff and subspace migration has been consistently studied thanks to its better imaging product. However, the existing research on this algorithm has been applied heuristically. That is, it mostly relied on experimental results. Although research on several structures was conducted, it was based on experiments or statistical approach, not revealing the mathematical structures explicitly, refer to \cite{AGKPS,BPT,D,DJRBSM,HHSZ,MGAD,P2,P3,PP,SSS} and references therein. It resulted in the difficulties of explaining the results theoretically.

Recently, an analysis of mathematical structure of single- and multi-frequency subspace migration for imaging of small electromagnetic materials has been conducted by establishing a relationship with Bessel function of integer order in full-view inverse scattering
problems, refer to \cite{JKHP1}. This remarkable research has shown the reason why subspace migration is effective and why
application of multi-frequency guarantees better imaging performance than application of single-frequency. Motivated by this work, this
analysis is successfully extended to the limited-view inverse scattering problems (see \cite{KP}).

Afterwards, a multi-frequency subspace migration weighted by applied frequency has been suggested in \cite{JKHP2} to obtain more precise results. Furthermore, the reason why the suggested algorithm presents better imaging products was demonstrated mathematically and an
analysis of multi-frequency subspace migration weighted by the power of applied frequency has been considered in \cite{P1}. This research
concludes that increasing the power of applied frequency is meaningless, so that multi-frequency subspace migration weighted by applied
frequency is a good algorithm for imaging.

Recently, it has been confirmed that a multi-frequency subspace migration weighted by the logarithmic function of applied frequency can yield more appropriate imaging result than the one suggested in \cite{JP}. However, one can face difficulties in identifying the reason why it shows the better performance through the mathematical analysis. Motivated by this difficulty, we derive an indefinite integration of square of Bessel function of order zero of the first kind multiplied by the natural logarithmic function. Based on this integration, we
discover the structure of multi-frequency subspace migration weighted by the logarithmic function of applied frequency by establishing a
relationship with Bessel function of integer order of the first kind, and provide the reason of better imaging performance.

The organization of this study is as follows. In Section \ref{Sec2}, we briefly introduce two-dimensional direct scattering problem and
subspace migration. Section \ref{Sec3} provides a survey on the structures of single-, multi-, and weighted multi-frequency subspace
migrations, the derivation of indefinite integration of square of Bessel function multiplied by the natural logarithmic function, and the
the mathematical analysis on why multi-frequency subspace migration weighted by natural logarithmic function shows better imaging performance than the traditional one. In Section \ref{Sec4}, several results of numerical experiments with noisy data are presented in order to support our analysis. Finally, a short conclusion is mentioned in Section \ref{Sec5}.

\section{Two-dimensional direct scattering problem and subspace migration}\label{Sec2}
\subsection{Direct scattering problem and asymptotic expansion formula}
Let us consider two-dimensional electromagnetic scattering from a thin, curve-like homogeneous inclusion within a homogeneous space
$\mathbb{R}^2$. The latter contains an inclusion denoted as $\Gamma$ which is localized in the neighborhood of a curve $\sigma$. That is,
\begin{equation}\label{TI}
\Gamma=\set{\mx+\eta\mn(\mx):\mx\in\sigma,~\eta\in(-h,h)},
\end{equation}
where the supporting $\sigma$ is a simple, smooth curve in $\mathbb{R}^2$, $\mn(\mx)$ is the unit normal to $\sigma$ at $\mx$, and $h$ is a
strictly positive constant which specifies the thickness of the inclusion (small with respect to the wavelength), refer to figure
\ref{ThinInclusion}. Throughout this paper, we denote $\mt(\mx)$ be the unit tangent vector at $\mx\in\sigma$.

\begin{figure}
\begin{center}
\includegraphics[width=0.35\textwidth,keepaspectratio=true,angle=0]{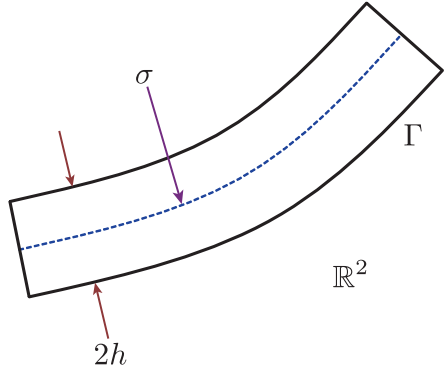}
\caption{\label{ThinInclusion}Sketch of the thin inclusion $\Gamma$ in two-dimensional space $\mathbb{R}^2$.}
\end{center}
\end{figure}

In this paper, we assume that every material is characterized by its dielectric permittivity and magnetic permeability at a given
frequency. Let $0<\eps_0<+\infty$ and $0<\mu_0<+\infty$ denote the permittivity and permeability of the embedding space $\mathbb{R}^2$, and
$0<\eps<+\infty$ and $0<\mu<+\infty$ the ones of the inclusion $\Gamma$. Then, we can define the following piecewise constant dielectric
permittivity
\begin{equation}\label{EP}
\eps(\mx)=\left\{\begin{array}{ccl}
\eps_0&\mbox{for}&\mx\in\mathbb{R}^2\backslash\overline{\Gamma}\\
\eps&\mbox{for}&\mx\in\Gamma
\end{array}\right.
\end{equation}
and magnetic permeability
\begin{equation}\label{MP}
\mu(\mx)=\left\{\begin{array}{ccl}
\mu_0&\mbox{for}&\mx\in\mathbb{R}^2\backslash\overline{\Gamma}\\
\mu&\mbox{for}&\mx\in\Gamma,
\end{array}\right.
\end{equation}
respectively. Note that if there is no inclusion, i.e., in the homogeneous space, $\mu(\mx)$ and $\eps(\mx)$ are equal to $\mu_0$ and
$\eps_0$ respectively. In this paper, we set $\eps>\eps_0=1$ and $\mu>\mu_0=1$ for convenience.

At strictly positive operation frequency $\omega$ (wavenumber $k_0=\omega\sqrt{\eps_0\mu_0}=\omega$), let $u_{\mathrm{tot}}^{(l)}(\mx;\omega)$ be the time-harmonic total field which satisfies the Helmholtz equation
\begin{equation}\label{HE1}
\nabla\cdot\left(\frac{1}{\mu(\mx)}\nabla u_{\mathrm{ tot}}^{(l)}(\mx;\omega)\right)+\omega^2\eps(\mx)u_{\mathrm{
tot}}^{(l)}(\mx;\omega)=0\quad\mbox{in}\quad\mathbb{R}^2.
\end{equation}
Similarly, the incident field $u_{\mathrm{back}}^{(l)}(\mx;\omega)$ satisfies the homogeneous Helmholtz equation
\[\nabla\cdot\left(\frac{1}{\mu_0}\nabla u_{\mathrm{back}}^{(l)}(\mx;\omega)\right)+\omega^2
\eps_0u_{\mathrm{back}}^{(l)}(\mx;\omega)=0\quad\mbox{in}\quad\mathbb{R}^2.\]
As is usual, the total field $u_{\mathrm{tot}}^{(l)}(\mx;\omega)$ divides itself into the incident field $u_{\mathrm{back}}^{(l)}(\mx;\omega)$ and the scattered field $u_s$, $u=u_0+u_s$. Notice that this unknown
scattered field $u_{\mathrm{scat}}^{(l)}(\mx;\omega)$ satisfies the Sommerfeld radiation condition
\[\lim_{|\mx|\to\infty}\sqrt{|\mx|}\left(\frac{\p u_{\mathrm{ scat}}^{(l)}(\mx;\omega)}{\p|\mx|}-i\omega
u_{\mathrm{scat}}^{(l)}(\mx;\omega)\right)=0\]
uniformly in all directions $\hat{\mx}=\mx/|\mx|$.

In this paper, we consider the illumination of plane waves
\[u_{\mathrm{back}}^{(l)}(\mx;\omega)=e^{i\omega\vt_l\cdot\mx}\quad\mbox{for}\quad x\in\mathbb{R}^2\]
and far fields in free space where $\vt_l$ is a two-dimensional vector on the unit circle $\mathbb{S}^1$ in $\mathbb{R}^2$. For convenience, we denote $\set{\vv_j:j=1,2,\cdots,N}\subset\mathbb{S}^1$ to be a discrete finite set of observation directions and
$\set{\vt_l:l=1,2,\cdots,N}\subset\mathbb{S}^1$ be the set of incident directions.

The far-field pattern is defined as a function $u_{\infty}(\vv,\vt;\omega)$ which satisfies \[u_{\infty}(\vv,\vt;\omega)=\frac{e^{i\omega|\my|}}{\sqrt{|\my|}}u_{\mathrm{scat}}^{(l)}(\mx;\omega)+o\left(\frac{1}{\sqrt{|\my|}}\right)\]
as $|\my|\longrightarrow\infty$ uniformly on $\vv=\my/|\my|\in\mathbb{S}^1$ and $\vt\in\mathbb{S}^1$. Then, based on \cite{BF},
$u_{\infty}(\vv,\vt;\omega)$ can be written as an asymptotic expansion formula.

\begin{lem}[See \cite{BF}]
  For $\vv,\vt\in\mathbb{S}^1$ and $\mx\in\mathbb{R}^2\backslash\overline{\Gamma}$, the far-field pattern $u_{\infty}^{(l)}(\mx,\vt;\omega)$
  can be represented as
  \[u_{\infty}(\vv,\vt;\omega)=h\frac{\omega^2(1+i)}{4\sqrt{\omega\pi}}\int_{\sigma}\bigg((\eps-1)-\vv\cdot\mathbb{M}(\my)\cdot\vt\bigg)e^{i\omega(\vt-\vv)\cdot\my}d\sigma(\my),\]
  where $o(h)$ is uniform in $\my\in\sigma$, $\vv,\vt\in\mathbb{S}^1$, and $\mathbb{M}(\my)$ is a $2\times2$ symmetric matrix defined as
  follows: let $\mt(\my)$ and $\mn(\my)$ denote unit tangent and normal vectors to $\sigma$ at $\my$, respectively. Then
    \begin{itemize}
      \item $\mathbb{M}(\my)$ has eigenvectors $\mt(\my)$ and $\mn(\my)$.
      \item The eigenvalue corresponding to $\mt(\my)$ is $2\left(\frac{1}{\mu}-\frac{1}{\mu_0}\right)=2\left(\frac{1}{\mu}-1\right)$.
      \item The eigenvalue corresponding to $\mn(\my)$ is $2\left(\frac{1}{\mu_0}-\frac{\mu}{\mu_0^2}\right)=2(1-\mu)$.
    \end{itemize}
\end{lem}

\subsection{Introduction to subspace migration}
Now, we introduce subspace migration for imaging of thin inclusion $\Gamma$. Detailed description can be found in \cite{ABC,AGKPS,P1,P3,PL2}. In
order to introduce, we generate a Multi-Static Response (MSR) matrix $\mathbb{K}(\omega)=[K_{jl}(\omega)]_{j,l=1}^{N}\in\mathbb{C}^{N\times
N}$ whose element $K_{jl}=u_{\infty}(\vv_j,\vt_l;\omega)$ is the collected far-field at observation number $j$ for the incident number $l$.
In this paper, we assume that $\vv_j=-\vt_j$, i.e., we have the same incident and observation directions configuration. It is worth
emphasizing that for a given frequency $\omega=\frac{2\pi}{\lambda}$, based on the resolution limit, any detail less than one-half of the
wavelength cannot be retrieved. Hence, if we divide thin inclusion $\Gamma$ into $M$ different segments of size of order
$\frac{\lambda}{2}$, only one point, say, $\my_m$, $m=1,2,\cdots,M$, at each segment will affect the imaging (see \cite{ABC,AGKPS,PL1,PL3}).
If $M<N$, the elements of MSR matrix can be represented as follows:
\begin{align}
\begin{aligned}\label{ElementofMSR}
  u_{\infty}(\vv_j,\vt_l;\omega)=&h\frac{\omega^2(1+i)}{4\sqrt{\omega\pi}}\int_{\sigma}\bigg((\eps-1)-\vv\cdot\mathbb{M}(\my)\cdot\vt\bigg)e^{i\omega(\vt-\vv)\cdot\my}d\sigma(\my)\\
  =&h\frac{\omega^2(1+i)}{4\sqrt{\omega\pi}}\frac{|\sigma|}{M}\sum_{m=1}^{M}\bigg[(\eps-1)+2\left(\frac{1}{\mu}-1\right)
  \vt_j\cdot\mt(\my_m)\vt_l\cdot\mt(\my_m)\\
  &+2\left(1-\mu\right)\vt_j\cdot\mn(\my_m)\vt_l\cdot\mn(\my_m)\bigg]e^{i\omega(\vt_j+\vt_l)\cdot\my_m},
\end{aligned}
\end{align}
where $|\sigma|$ denotes the length of $\sigma$.

Now, let us perform the Singular Value Decomposition (SVD) of $\mathbb{K}(\omega)$
\[\mathbb{K}(\omega)=\mathbb{U}(\omega)\mathbb{S}(\omega)\overline{\mathbb{V}}(\omega)^T=\sum_{m=1}^{N}\rho_m(\omega)\mU_m(\omega)\overline{\mV}_m(\omega)^T\approx\sum_{m=1}^{M}\rho_m(\omega)\mU_m(\omega)\overline{\mV}_m(\omega)^T,\]
where $\rho_m(\omega)$, $m=1,2,\cdots,N$ nonzero singular values such that
\[\rho_1(\omega)\geq\rho_2(\omega)\geq\cdots\geq\rho_M(\omega)>0\quad\mbox{and}\quad\rho_m(\omega)\approx0\quad\mbox{for}\quad
m=M+1,M+2,\cdots,N,\]
and $\mU_m(\omega)$ and $\mV_m(\my_m;\omega)$ are left- and right-singular vectors of $\mathbb{K}(\omega)$, respectively.

Based on the structure of (\ref{ElementofMSR}), define a vector $\mW(\mz;\omega)\in\mathbb{C}^{N\times1}$ as
\begin{equation}\label{VecW}
  \mW(\mz;\omega)=\bigg[\mc\cdot[1,\vt_1]e^{i\omega\vt_1\cdot\mz},\mc\cdot[1,\vt_2]e^{i\omega\vt_2\cdot\mz},\cdots,
  \mc\cdot[1,\vt_N]e^{i\omega\vt_N\cdot\mz}\bigg]^T,
\end{equation}
where the selection of $\mc\in\mathbb{R}^3\backslash\set{\mathbf{0}}$ depends on the shape of the supporting curve $\sigma)$ (see
\cite[Section 4.3.1]{PL3} for a detailed discussion). Then, in accordance with \cite{AGKPS},
\[\mW(\my_m;\omega)=e^{i\gamma_m^1}\mU_m(\omega)\quad\mbox{and}\quad\mW(\my_m;\omega)=e^{i\gamma_m^2}\overline{\mV}_m(\omega)\]
for some $\gamma_m^1$ and $\gamma_m^2$, $m=1,2,\cdots,M$. Based on the orthonormal property of singular vectors, the first $M$ columns of
$\mathbb{U}(\omega)$ and $\mathbb{V}(\omega)$ are orthonormal, it follows that
\begin{align}
\begin{aligned}\label{orthonormal}
  &\langle\mW(\mz;\omega),\mU_m(\omega)\rangle\approx1,\quad\langle\mW(\mz;\omega),\overline{\mV}_m(\omega)\rangle\approx1\quad\mbox{if}\quad\mz=\my_m\\
  &\langle\mW(\mz;\omega),\mU_m(\omega)\rangle\approx0,\quad\langle\mW(\mz;\omega),\overline{\mV}_m(\omega)\rangle\approx0\quad\mbox{if}\quad\mz\ne\my_m,
\end{aligned}
\end{align}
where $\langle\mathbf{a},\mathbf{b}\rangle=\overline{\mathbf{a}}\cdot\mathbf{b}$.

Hence, we can introduce subspace migration for imaging of thin inclusion at a given frequency $\omega$ as
\begin{equation}\label{SingleSM}
  \W_{\mathrm{SF}}(\mz;\omega):=\left|\sum_{m=1}^{M}\langle\mW(\mz;\omega),\mU_m(\omega)\rangle\langle\mW(\mz;\omega),\overline{\mV}_m(\omega)\rangle\right|.
\end{equation}
Based on the properties (\ref{orthonormal}), map of $\W_{\mathrm{SF}}(\mz;\omega)$ should exhibit peaks of magnitude $1$ at
$\mz=\my_m\in\sigma$, and of small magnitude at $\mz\in\mathbb{R}^2\backslash\overline{\Gamma}$.

\section{Multi-frequency subspace migration weighted by natural logarithm function}\label{Sec3}
\subsection{Structure of single- and multi-frequency subspace migration}\label{Sec3-1}
In recent work \cite{JKHP1}, the structure of $\W_{\mathrm{SF}}(\mz;\omega)$ is derived as follows. Throughout this paper, we assume that
the set of incident (and also observation) directions $\{\vt_j:j=1,2,\cdots,N\}$ spans $\mathbb{S}^1$, and $J_n(x)$ be the \textit{Bessel
function of integer order $n$ of the first kind}.

\begin{lem}[See \cite{JKHP1}]
  If the total number of incident and observation directions $N$ is sufficiently large and satisfies $N>3M$. Then single-frequency subspace
  migration (\ref{SingleSM}) can be represented as follows:
  \[\W_{\mathrm{SF}}(\mz;\omega)\approx\sum_{m=1}^{M}J_0(\omega|\mz-\my_m|)^2.\]
\end{lem}
This result tells us that although the shape of $\Gamma$ can be recognized via the map of $\W_{\mathrm{SF}}(\mz;\omega)$, some unexpected
artifacts should appear in the map of $\W_{\mathrm{SF}}(\mz;\omega)$ due to the oscillating property of Bessel function.

In order to obtain better result, a normalized multi-frequency subspace migration is considered. This is introduced as follows; for
multi-frequency $\set{\omega_f:f=1,2,\cdots,F}$, by an assumption of $M_f=M$ for all $f$, a normalized multi-frequency subspace migration
$\W_{\mathrm{MF}}(\mz;F)$ is given by
\begin{equation}\label{MultiSM}
  \W_{\mathrm{MF}}(\mz;F):=\left|\sum_{f=1}^{F}\sum_{m=1}^{M_f}\langle\mW(\mz;\omega_f),\mU_m(\omega_f)\rangle\langle\mW(\mz;\omega_f),\overline{\mV}_m(\omega_f)\rangle\right|,
\end{equation}
where $M_f$ is the number of nonzero singular values of MSR matrix $\mathbb{K}(\omega_f)$, $f=1,2,\cdots,F$. Then, the structure of
(\ref{MultiSM}) can be represented as follows.
\begin{lem}[See \cite{JKHP1}]
  If the total number of incident and observation directions $N$ is sufficiently large and satisfies $N>3M_f$ for $f=1,2,\cdots,F$. Then,
  multi-frequency subspace migration (\ref{MultiSM}) can be represented as follows:
  \begin{multline}\label{StructureMultiSM}
    \W_{\mathrm{MF}}(\mz;F)\approx
    F\sum_{m=1}^{M}\left\{\frac{\omega_F}{\omega_F-\omega_1}\bigg(J_0(\omega_F|\mz-\my_m|)^2+J_1(\omega_F|\mz-\my_m|)^2\bigg)\right.\\
    \left.-\frac{\omega_1}{\omega_F-\omega_1}\bigg(J_0(\omega_1|\mz-\my_m|)^2+J_1(\omega_1|\mz-\my_m|)^2\bigg)+\int_{\omega_1}^{\omega_F}J_1(\omega|\mz-\my_m|)^2d\omega\right\}.
  \end{multline}
\end{lem}

This shows that (\ref{MultiSM}) yields better images owing to less oscillation than (\ref{SingleSM}) does so that unexpected artifacts in
the plot of $\W_{\mathrm{MF}}(\mz;F)$ are mitigated when $F$ is sufficiently large. This result indicates why a multi-frequency subspace
migration offers images with good resolution. On the other hand, it is possible to examine the same conclusion throughout the basis of
Statistical Hypothesis Testing in Statistical theory, refer to \cite{AGKPS,HHSZ}.

In order to improve multi-frequency subspace migration (\ref{MultiSM}), one must eliminate or control the last term of
(\ref{StructureMultiSM}). For this purpose, a weighted multi-frequency subspace migration has been introduced
\begin{equation}\label{WeightedMultiSM}
  \W_{\mathrm{WMF}}(\mz;F,n):=\left|\sum_{f=1}^{F}\sum_{m=1}^{M_f}(\omega_f)^n\langle\mW(\mz;\omega_f),\mU_m(\omega_f)\rangle\langle\mW(\mz;\omega_f),\overline{\mV}_m(\omega_f)\rangle\right|,
\end{equation}
and following result is obtained.

\begin{lem}[See \cite{P1}]
  If the total number of incident and observation directions $N$ is sufficiently large and satisfies $N>3M_f$ for $f=1,2,\cdots,F$. Then,
  multi-frequency subspace migration (\ref{WeightedMultiSM}) can be represented as follows:
  \begin{multline}\label{StructureWeightedMultiSM}
    \W_{\mathrm{WMF}}(\mz;F,n)\approx\frac{F}{n+1}\sum_{m=1}^{M}\left\{\frac{(\omega_F)^{n+1}}{\omega_F-\omega_1}\bigg(J_0(\omega_F|\mz-\my_m|)^2+J_1(\omega_F|\mz-\my_m|)^2\bigg)\right.\\
    \left.-\frac{(\omega_1)^{n+1}}{\omega_F-\omega_1}\bigg(J_0(\omega_1|\mz-\my_m|)^2+J_1(\omega_1|\mz-\my_m|)^2\bigg)+\mathbb{D}(|\mz-\my_m|,\omega_1,\omega_F;n)\right\}.
  \end{multline}
\end{lem}

Based on recent work \cite{P1}, the term $\mathbb{D}(|\mz-\my_m|,\omega_1,\omega_F;n)$, which is disturbing the imaging performance, is
eliminated when $n=1$ and remained when $n=0,2,3,\cdots$, respectively. Hence, it can be said that $\W_{\mathrm{WMF}}(\mz;F,1)$ is an
improved version of $\W_{\mathrm{MF}}(\mz;F)$.

\subsection{Structure of multi-frequency subspace migration weighted by natural logarithmic function}
Based on the improved procedure presented above, it is natural to consider the following multi-frequency subspace migration weighted by some
function $\xi(\omega)$:
\begin{equation}\label{WMSM}
  \W(\mz;F):=\left|\sum_{f=1}^{F}\sum_{m=1}^{M_f}\xi(\omega_f)\langle\mW(\mz;\omega_f),\mU_m(\omega_f)\rangle\langle\mW(\mz;\omega_f),\overline{\mV}_m(\omega_f)\rangle\right|.
\end{equation}
Then, based on the results in Section \ref{Sec3-1}, $\W(\mz;F)$ should be represented as the following form:
\begin{multline*}
  \W(\mz;F)=C\sum_{m=1}^{M}\left\{\frac{\phi(\omega_F)}{\omega_F-\omega_1}\bigg(J_0(\omega_F|\mz-\my_m|)^2+J_1(\omega_F|\mz-\my_m|)^2\bigg)\right.\\
  \left.-\frac{\phi(\omega_1)}{\omega_F-\omega_1}\bigg(J_0(\omega_1|\mz-\my_m|)^2+J_1(\omega_1|\mz-\my_m|)^2\bigg)+\mathbb{E}(|\mz-\my_m|,\omega_1,\omega_F;n)\right\},
\end{multline*}
where $\phi$ is a positive definite function. Based on several recent works \cite{JKHP1,P1}, the term
$\mathbb{E}(|\mz-\my_m|,\omega_1,\omega_F;n)$ is disturbing the imaging performance since it generates some unnecessary artifacts. However,
if we can find a suitable function $\xi(\omega_f)$, which makes $\mathbb{E}(|\mz-\my_m|,\omega_1,\omega_F;n)$ as a negative valued function
in the neighborhood of $\Gamma$, and small (or negative) valued one at the outside of neighborhood of $\Gamma$, $\W(\mz;F)$ will be an
improved subspace migration. Recently, it has been confirmed that $\W(\mz;F)$ is an improved version of $\W_{\mathrm{WMF}}(\mz;F,1)$ when
$\xi(\omega_f)=\ln(\omega_f)$. However, this fact has been examined via some results of numerical simulations. Thus, identification of
mathematical structure of (\ref{WMSM}) is still remaining. In order to identify this, we derive the following indefinite integration. This plays
an important role in exploring the structure of (\ref{WMSM}). Note that the derivation of following integration is easy but we have been
unable to find such a derivation.

\begin{thm}\label{TheoremLogIntegral}
  For every positive real number $x$, following identity holds
  \begin{equation}\label{IntegralFormulaLog}
    \int \ln(x)J_0(x)^2dx=(x\ln(x)-x)\left(J_0(x)^2+J_1(x)^2\right)+\int(\ln(x)-2)J_1(x)^2dx.
  \end{equation}
\end{thm}
\begin{proof}
Since
\[\frac{d}{dx}J_0(x)=-J_1(x),\]
an integration by parts yields
\begin{equation}\label{term1}
  \int\ln(x)J_0(x)^2dx=(x \ln(x)-x)J_0(x)^2 +2\int x(\ln(x)-1)J_0(x)J_1(x)dx.
\end{equation}
Applying following identity (see \cite[page 17]{R})
\[\int x^2 J_0(x)J_1(x)dx=\frac{x^2}{2}J_1(x)^2,\]
leads us
\begin{align}
\begin{aligned}\label{term2}
  \int x(\ln(x)-1)J_0(x)J_1(x)dx&=\int\left(\frac{\ln(x)-1}{x}\right)x^2 J_0(x)J_1(x)dx\\
  &=\left(\frac{\ln(x)-1}{x}\right)\frac{x^2}{2}J_1(x)^2-\int\left(\frac{2-\ln(x)}{x^2}\right)\frac{x^2}{2}J_1(x)^2dx\\
  &=\frac12\left((x\ln(x)-x)J_1(x)^2-\int(2-\ln(x))J_1(x)^2 dx\right).\\
\end{aligned}
\end{align}
Combining (\ref{term1}) and (\ref{term2}), we can obtain the desired result (\ref{IntegralFormulaLog}). This completes the proof.
\end{proof}

Applying Theorem \ref{TheoremLogIntegral}, we can explore the structure of $\W(\mz;F)$.
\begin{thm}\label{StructureWMSM}
  Assume that total number of incident and observation directions $N$ is sufficiently large and satisfies $N>3M_f$ for $f=1,2,\cdots,F$.
  Then, multi-frequency subspace migration (\ref{WMSM}) weighted by $\xi(\omega_f)=\ln(\omega_f)$ can be represented as follows:
  \begin{align*}\label{StructureWeightedMultiSM}
    \W(\mz;F)\approx
    \frac{F}{\omega_F-\omega_1}\sum_{m=1}^{M_f}&\left\{\omega_F\ln(\omega_F)\bigg(J_0(\omega_F|\mz-\my_m|)^2+J_1(\omega_F|\mz-\my_m|)^2\bigg)\right.\\
    &-\omega_1\ln(\omega_1)\bigg(J_0(\omega_1|\mz-\my_m|)^2+J_1(\omega_1|\mz-\my_m|)^2\bigg)\\
    &\left.-\int_{\omega_1}^{\omega_F}\bigg(J_0(\omega|\mz-\my_m|)^2-(\ln(\omega)-1)J_1(\omega|\mz-\my_m|)^2\bigg)d\omega\right\}.
  \end{align*}
\end{thm}
\begin{proof}
  For the sake of simplicity, we assume that $M_f=M$ for all $f=1,2,\cdots,F$, and denote
  \[\Lambda(x):=J_0(x)^2 +J_1(x)^2.\]
  Then by virtue of \cite{JKHP1}, (\ref{StructureWeightedMultiSM}) can be written as follows
  \[\W(\mz;F)\approx\left|\frac{F}{\omega_F-\omega_1}\sum_{m=1}^{M}\int_{\omega_1}^{\omega_F}\ln(\omega)J_0(\omega|\mz-\my_m|)^2\right|.\]
  Then, the change of variable $x=\omega|\mz-\my_m|$ yields
  \begin{multline*}
    \int_{\omega_1}^{\omega_F}\ln(\omega)J_0(\omega|\mz-\my_m|)^2=\frac{1}{|\mz-\my_m|}\int_{\omega_1|\mz-\my_m|}^{\omega_F|\mz-\my_m|}\ln\left(\frac{x}{|\mz-\my_m|}\right)J_0(x)^2dx\\
    =\frac{1}{|\mz-\my_m|}\int_{\omega_1|\mz-\my_m|}^{\omega_F|\mz-\my_m|}\ln(x)J_0(x)^2dx-\frac{\ln|\mz-\my_m|}{|\mz-\my_m|}\int_{\omega_1|\mz-\my_m|}^{\omega_F|\mz-\my_m|}J_0(x)^2dx.
  \end{multline*}

  Since
  \begin{multline*}
    \frac{1}{|\mz-\my_m|}\int_{\omega_1|\mz-\my_m|}^{\omega_F|\mz-\my_m|}\ln(x)J_0(x)^2dx=\\
    \frac{1}{|\mz-\my_m|}\left\{\bigg[\bigg(x\ln(x)-x\bigg)\Lambda(x)\bigg]_{\omega_1|\mz-\my_m|}^{\omega_F|\mz-\my_m|}
    -\int_{\omega_1|\mz-\my_m|}^{\omega_F|\mz-\my_m|}(2-\ln(x))J_1(x)^2dx\right\},
  \end{multline*}
  applying Theorem \ref{TheoremLogIntegral} yields
  \begin{multline}\label{term3}
    \frac{1}{|\mz-\my_m|}\int_{\omega_1|\mz-\my_m|}^{\omega_F|\mz-\my_m|}\ln(x)J_0(x)^2dx=\omega_F\bigg(\ln(\omega_F|\mz-\my_m|)-1\bigg)\Lambda(\omega_F|\mz-\my_m|)\\
    -\omega_1\bigg(\ln(\omega_1|\mz-\my_m|)-1\bigg)\Lambda(\omega_1|\mz-\my_m|)+\int_{\omega_1}^{\omega_F}(\ln(\omega|\mz-\my_m|)-2)J_1(\omega|\mz-\my_m|)^2d\omega.
  \end{multline}

  Finally, let us apply following indefinite integral formula of the Bessel function (see \cite[page 35]{R}):
  \begin{equation}\label{BesselFormula}
    \int J_0(x)^2dx=x\bigg(J_0(x)^2+J_1(x)^2\bigg)+\int J_1(x)^2dx.
  \end{equation}
  Then, an elementary calculus leads us to
  \begin{multline}\label{term4}
    \frac{\ln|\mz-\my_m|}{|\mz-\my_m|}\int_{\omega_1|\mz-\my_m|}^{\omega_F|\mz-\my_m|}J_0(x)^2dx=\omega_F\ln|\mz-\my_m|\Lambda(\omega_F|\mz-\my_m|)\\
    -\omega_1\ln|\mz-\my_m|\Lambda(\omega_1|\mz-\my_m|)+\ln|\mz-\my_m|\int_{\omega_1}^{\omega_F}J_1(\omega|\mz-\my_m|)^2d\omega.
  \end{multline}

  By combining (\ref{term3}) and (\ref{term4}), we can obtain
  \begin{multline*}
    \int_{\omega_1}^{\omega_F}\ln(\omega)J_0(\omega|\mz-\my_m|)^2=
    \frac{F}{\omega_F-\omega_1}\left\{(\omega_F\ln(\omega_F)-\omega_F)\bigg(J_0(\omega_F|\mz-\my_m|)^2+J_1(\omega_F|\mz-\my_m|)^2\bigg)\right.\\
    \left.-(\omega_1\ln(\omega_1)-\omega_1)\bigg(J_0(\omega_1|\mz-\my_m|)^2+J_1(\omega_1|\mz-\my_m|)^2\bigg)+\int_{\omega_1}^{\omega_F}(\ln(\omega)-2)J_1(\omega|\mz-\my_m|)^2d\omega\right\}.
  \end{multline*}
  With this, we can obtain the desired result by applying (\ref{BesselFormula}) again. This completes the proof.
\end{proof}

The above result leads us to the following result of improvement.
\begin{thm}\label{ImprovementTheorem}
   Under the same assumption of Theorem \ref{StructureWMSM}, multi-frequency subspace migration (\ref{WMSM}) weighted by
   $\xi(\omega_f)=\ln(\omega_f)$ improves $\W_{\mathrm{WMF}}(\mz;F,1)$ of (\ref{WeightedMultiSM}).
\end{thm}
\begin{proof}
  Throughout the proof, we assume that all $\omega_f=2\pi/\lambda_f$ are sufficiently large enough and $\lambda_1-\lambda_F$ are small
  enough such that\footnote{In Section \ref{Sec4}, we set $\lambda_1=0.5$ and $\lambda_F=0.3$. Therefore, the value $C$ is smaller than
  $1$.}
  \[0<\ln\left(\frac{\lambda_1}{\lambda_F}\right)<C.\]

  In order to show the improvement of (\ref{WMSM}), we recall the structure (\ref{StructureWeightedMultiSM})
  \begin{multline*}
    \W(\mz;F)\approx\frac{F}{\omega_F-\omega_1}\sum_{m=1}^{M_f}\left\{\omega_F\ln(\omega_F)\bigg(J_0(\omega_F|\mz-\my_m|)^2+J_1(\omega_F|\mz-\my_m|)^2\bigg)\right.\\
    \left.-\omega_1\ln(\omega_1)\bigg(J_0(\omega_1|\mz-\my_m|)^2+J_1(\omega_1|\mz-\my_m|)^2\bigg)-\E_1+\E_2\right\}.
  \end{multline*}
  where
  \[\E_1:=\int_{\omega_1}^{\omega_F}J_0(\omega|\mz-\my_m|)^2d\omega\quad\mbox{and}\quad
  \E_2:=\int_{\omega_1}^{\omega_F}(\ln(\omega)-1)J_1(\omega|\mz-\my_m|)^2d\omega.\]
  Based on the structure, it is enough to show that $\E_1\geq\E_2$, i.e, the term $-\E_1+\E_2$ is negative.

  \begin{enumerate}
  \item Suppose that $\mz$ is far away from $\my_m$, then since $\omega$ is sufficiently large, the following asymptotic form holds for any
      integer $n$,
  \begin{equation}\label{AsymptoticBessel}
    J_n(\omega|\mz-\my_m|)\approx\sqrt{\frac{2}{\omega|\mz-\my_m|}}\cos\left(\omega|\mz-\my_m|-\frac{n\pi}{2}-\frac{\pi}{4}\right)\longrightarrow0.
  \end{equation}
  Hence, $\E_1\longrightarrow0$. Moreover, applying L'H{\^o}pital's rule, it is easy to observe that
  \[(\ln(\omega)-1)J_1(\omega|\mz-\my_m|)^2\approx\frac{2(\ln(\omega)-1)}{\omega|\mz-\my_m|}\cos^2\left(\omega|\mz-\my_m|-\frac{3\pi}{4}\right)\longrightarrow0.\]
  Hence, we can conclude that the values $\E_1$ and $\E_2$ are negligible, i.e., $-\E_1+\E_2\approx0$.
  \item Assume that $\mz$ is close enough to $\my_m$ such that
  \[0<\omega_f|\mz-\my_m|\ll\sqrt{2}\]
  for all $f=1,2,\cdots,F$. Let $\alpha(x)$ denotes the Gamma function. Then, since the following asymptotic property of Bessel function
  holds for any integer $n$,
  \[J_n(x)\approx\frac{1}{\alpha(n+1)}\left(\frac{x}{2}\right)^n,\]
  we can observe that
  \[\E_1\approx\int_{\omega_1}^{\omega_F}d\omega=\omega_F-\omega_1=\mathcal{O}(\omega_F),\]
  and
  \begin{align*}
    \E_2&\approx\int_{\omega_1}^{\omega_F}(\ln(\omega)-1)\frac{\omega^2|\mz-\my_m|^2}{4}d\omega\ll\frac12\bigg[\omega\ln(\omega)-2\omega\bigg]_{\omega_1}^{\omega_F}\\
    &=\frac{1}{2}\bigg(\omega_F\ln(\omega_F)-2\omega_F-\omega_1\ln(\omega_1)+2\omega_1\bigg)\\
    &=\frac{1}{2}\left\{\omega_F\ln\left(\frac{\omega_F}{\omega_1}\right)+(\omega_F-\omega_1)\ln(\omega_1)-2(\omega_F-\omega_1)\right\}\\
    &<\frac{1}{2}\bigg(C\omega_F+\omega_F\ln(\omega_1)+2\omega_1\bigg)=\mathcal{O}(\omega_F).
  \end{align*}
  Therefore, the term $\E_2$ is dominated by $\E_1$, i.e., $-\E_1+\E_2$ is of negative value.
  \end{enumerate}
  Based on above observations, it is true that the term $-\E_1+\E_2$ is negative in the neighborhood of $\Gamma$, and close to zero at the
  outside of the neighborhood of $\Gamma$ (see Figure \ref{NegativeRegion}). This means that the results via (\ref{WMSM}) with
  $\xi(\omega)=\ln(\omega)$ will be better owing to less oscillation than (\ref{WeightedMultiSM}) with $n=1$ does. This completes the
  proof.
\end{proof}

\begin{figure}[!ht]
\begin{center}
\includegraphics[width=0.5\textwidth]{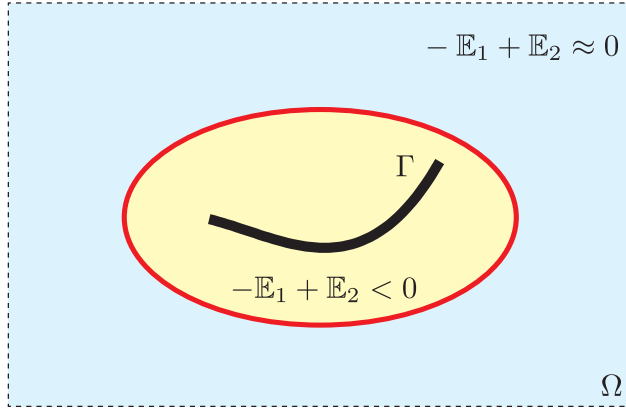}
\caption{\label{NegativeRegion}Description of the result in Theorem \ref{ImprovementTheorem}.}
\end{center}
\end{figure}

Note that $\W(\mz;F)$ has its maximum value at $\mz=\my_m\in\Gamma$. Hence, we can immediately obtain following result of unique
determination.
\begin{cor}
   Let the applied frequency $\omega$ be sufficiently high. If the total number $N$ of incident and observation directions and total number
   $F$ of applied frequencies are sufficiently large, then the shape of supporting curve $\sigma$ of thin inclusion $\Gamma$ can be obtained
   uniquely via the map of $\W(\mz;F)$.
\end{cor}

\section{Results of Numerical simulation and discussion}\label{Sec4}
\subsection{General configuration of numerical simulation}
Some numerical simulation experiments are performed in order to support Theorems \ref{StructureWMSM} and \ref{ImprovementTheorem}.
Throughout this section, the search vector $\mz$ is included in the square $\Omega=[-1,1]\times[-1,1]$. In order to describe thin inclusions
$\Gamma_j$, two smooth curves are selected as follows:
\begin{align*}
  \sigma_1&=\set{[s-0.2,-0.5s^2+0.5]^T~:~-0.5\leq s\leq0.5}\\
  \sigma_2&=\set{[s+0.2,s^3+s^2-0.6]^T~:~-0.5\leq s\leq0.5}.
\end{align*}
The thickness $h$ of thin inclusions $\Gamma_j$ is equally set to $0.015$. We denote $\eps_j$ and $\mu_j$ be the permittivity and
permeability of $\Gamma_j$, respectively, and set parameters $\mu_j$, $\mu_0$, $\eps_j$ and $\eps_0$ are $5,1,5$ and $1$, respectively.
Since $\mu_0$ and $\eps_0$ are set to unity, the applied frequencies reads as $\omega_f=2\pi/\lambda_f$ at wavelength $\lambda_f$ for
$f=1,2,\cdots,F(=10)$, which will be varied in the numerical examples between $\lambda_1=0.5$ and $\lambda_{10}=0.3$.

For the incident directions $\vt_l$, they are selected as
\[\vt_l=-\left[\cos\frac{2\pi(l-1)}{N},\sin\frac{2\pi(l-1)}{N}\right]^T\quad\mbox{for}\quad l=1,2,\cdots,N,\]
and $N=48$ total number of directions has chosen. Since $\mu>\mu_0$, the dominant eigenvectors of $\mathbb{K}(\omega_f)$ are $\mn(\my_m)$,
$\mc=[1,\mn(\my_m)]^T$ will be the best choice. However, we have no \textit{a priori} information of $\mn(\my_m)$, $\mc=[1,0,1]^T$ has
selected for generating $\mW(\mz;\omega_f)$ of (\ref{VecW}). For a more detailed discussion, we recommend a recent work \cite[Section
4.3]{PL3}.

In order to show the robustness, a white Gaussian noise with $10$dB signal-to-noise ratio (SNR) added to the unperturbed far-field pattern
data $u_{\infty}(\vv_j,\vt_l;\omega)$ via a standard MATLAB command \texttt{awgn} included in the \texttt{Communications System Toolbox}
package. For discriminating non-zero singular values of MSR matrix $\mathbb{K}(\omega_f)$, a $0.01-$threshold scheme is applied for each $f$ (see \cite{PL1,PL3} for instance). Throughout this section, only both permittivity and permeability contrast case is considered.

\subsection{Imaging results and related discussions}
In Figure \ref{Figure1}, some imaging results via $\W_{\mathrm{MF}}(\mz;10)$, $\W_{\mathrm{WMF}}(\mz;10,1)$, and $\W(\mz;10)$ are exhibited
when the thin inclusion is $\Gamma_1$. By comparing these results, we can immediately observe that unexpected artifacts can be examined in
the results via $\W_{\mathrm{MF}}(\mz;10)$ and $\W_{\mathrm{WMF}}(\mz;10,1)$ but, as we expected, they dramatically disappeared in the
map of $\W(\mz;10)$.

\begin{figure}[!ht]
\begin{center}
\includegraphics[width=0.325\textwidth]{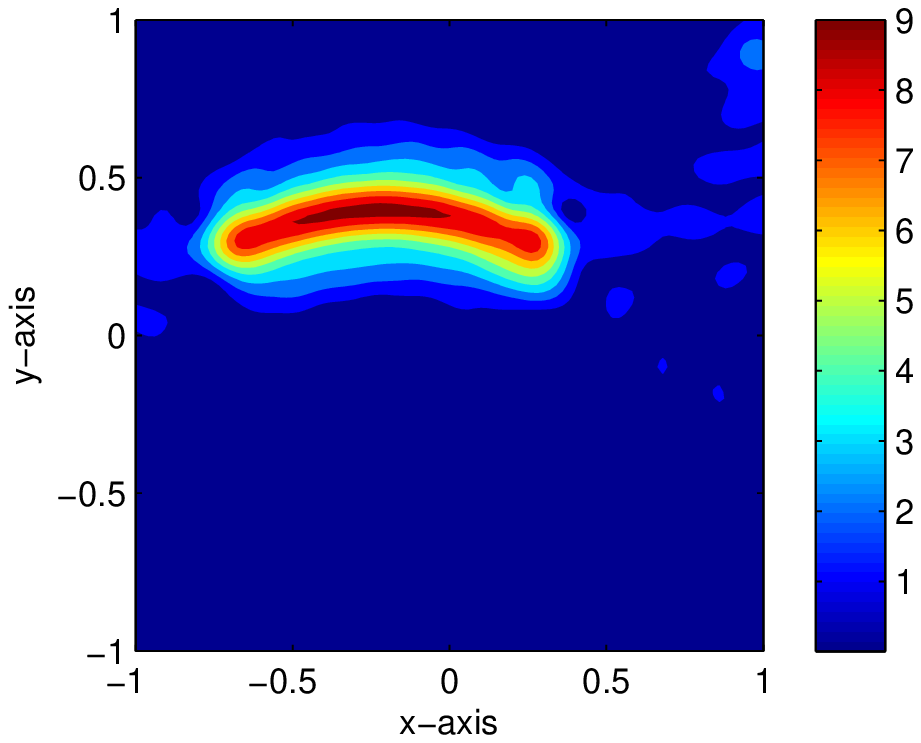}
\includegraphics[width=0.325\textwidth]{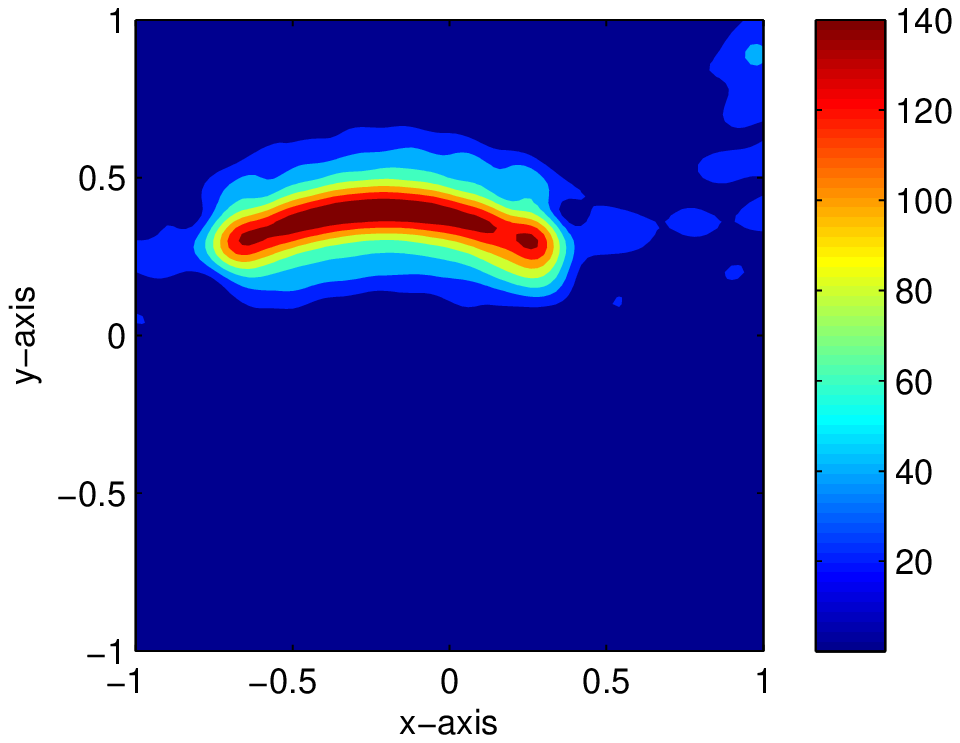}
\includegraphics[width=0.325\textwidth]{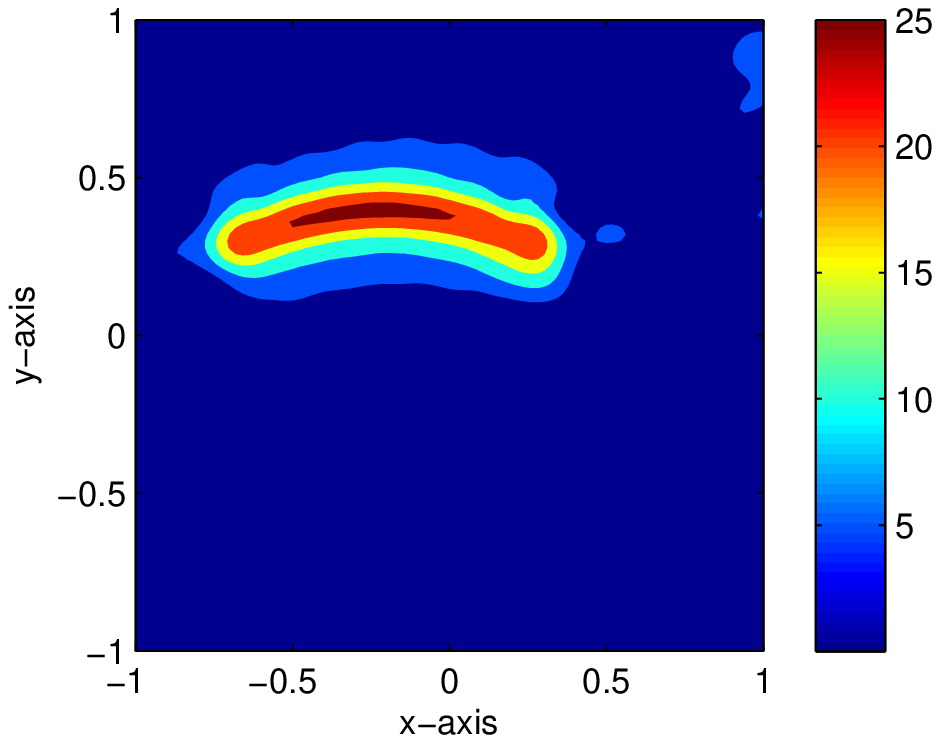}
\caption{\label{Figure1}Maps of $\W_{\mathrm{MF}}(\mz;10)$ (left), $\W_{\mathrm{WMF}}(\mz;10,1)$ (center), and $\W(\mz;10)$ (right) when the
thin inclusion is $\Gamma_1$.}
\end{center}
\end{figure}

Now, let us consider the imaging results of $\Gamma_2$ in Figure \ref{Figure2}. Based on this, we can observe the same phenomenon as in Figure \ref{Figure1}. Although proposed algorithm successfully eliminates artifacts, some part of $\Gamma_2$ can't be visible. This is due to the selection of $\mc=[1,0,1]^T$. Based on the shape of $\sigma_2$, the unit normal direction $\mn$ is similar to $[0,1]^T$ for $-0.5\leq s\leq0.3$. However, when $s\geq0.3$, $\mn$ is immensely different from $[0,1]^T$. Hence, finding an optimal $\mc$ is still remaining as an
interesting subject.

\begin{figure}[!ht]
\begin{center}
\includegraphics[width=0.325\textwidth]{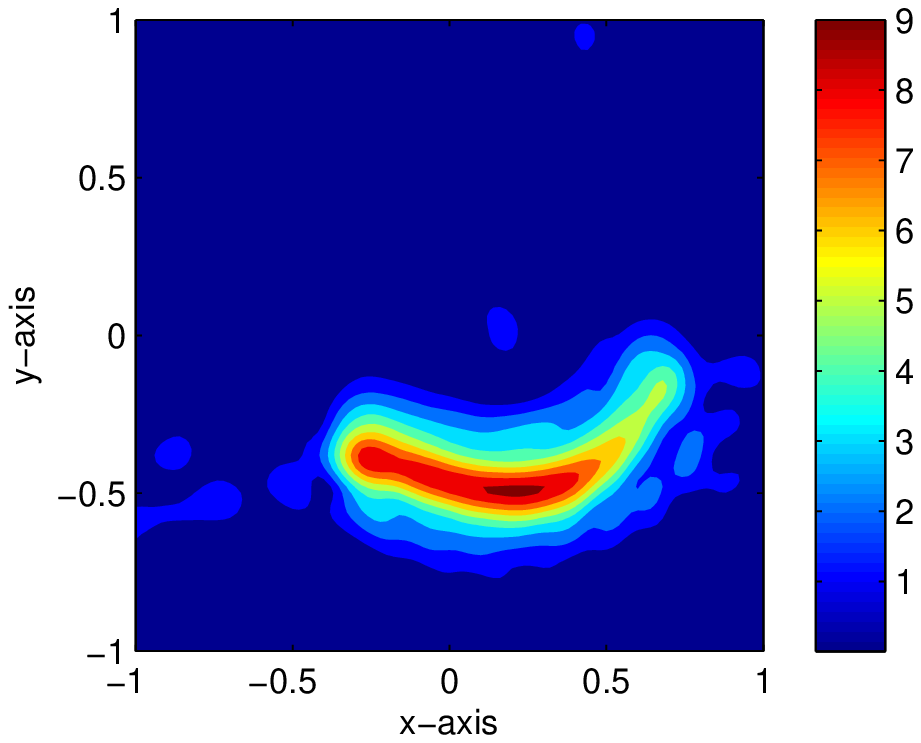}
\includegraphics[width=0.325\textwidth]{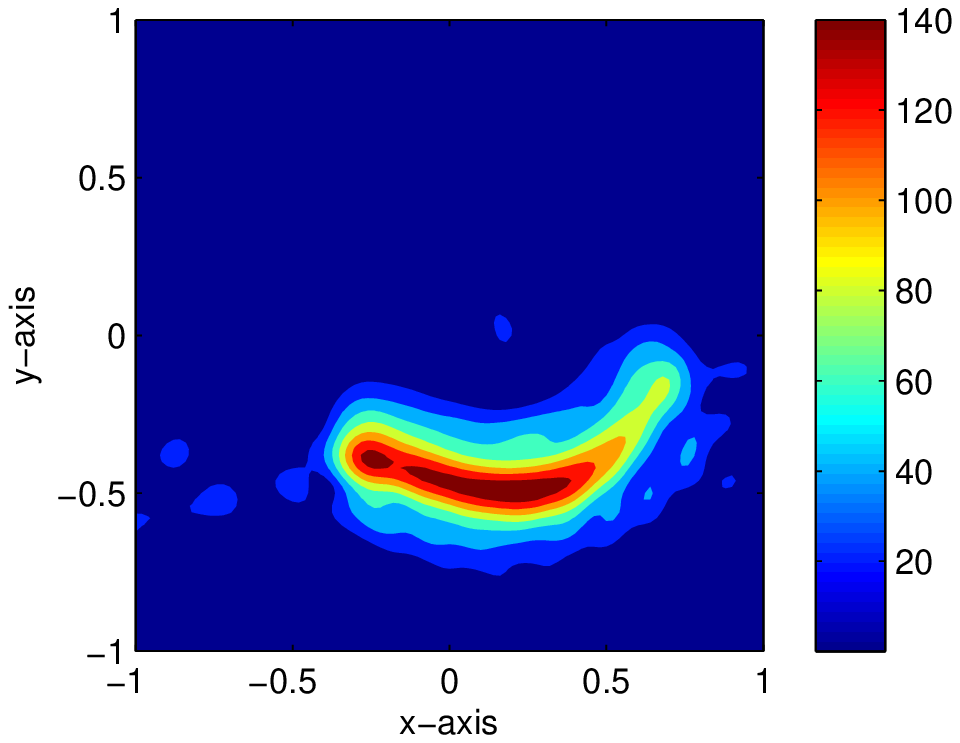}
\includegraphics[width=0.325\textwidth]{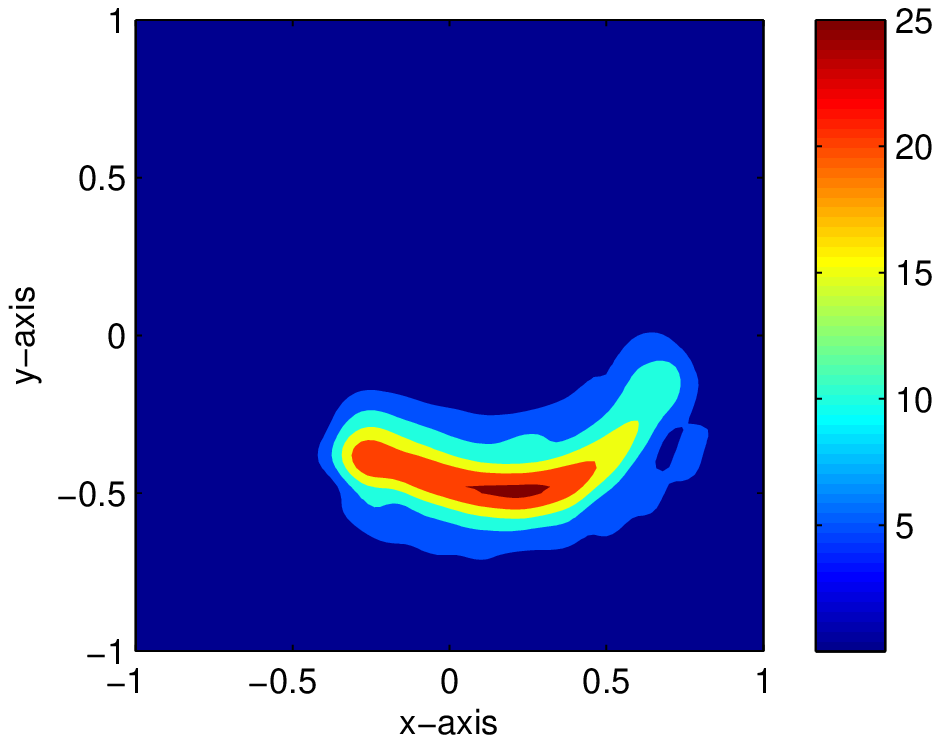}
\caption{\label{Figure2}Same as Figure \ref{Figure1} except the thin inclusion is $\Gamma_2$.}
\end{center}
\end{figure}

Now, let us extend the proposed algorithm for imaging of two-(or more) different thin inclusions $\Gamma_1$ and $\Gamma_2$. For the sake of
simplicity, we denote $\Gamma_{\mathrm{M}}=\Gamma_1\cup\Gamma_2$. Figure \ref{FigureM1} shows maps of $\W_{\mathrm{MF}}(\mz;10)$,
$\W_{\mathrm{WMF}}(\mz;10,1)$, and $\W(\mz;10)$ for $\Gamma_{\mathrm{M}}$ with the same permittivities $\eps_1=\eps_2=5$ and permeabilities
$\mu_1=\mu_2=5$. It is interesting to observe that unlike the case of single inclusion, where almost unexpected artifacts are eliminated,
proposed imaging functional $\W(\mz;F)$ does not improve the traditional ones $\W_{\mathrm{MF}}(\mz;F)$ and $\W_{\mathrm{WMF}}(\mz;F,1)$.
Hence, further analysis is needed to identify the reason.

\begin{figure}[!ht]
\begin{center}
\includegraphics[width=0.325\textwidth]{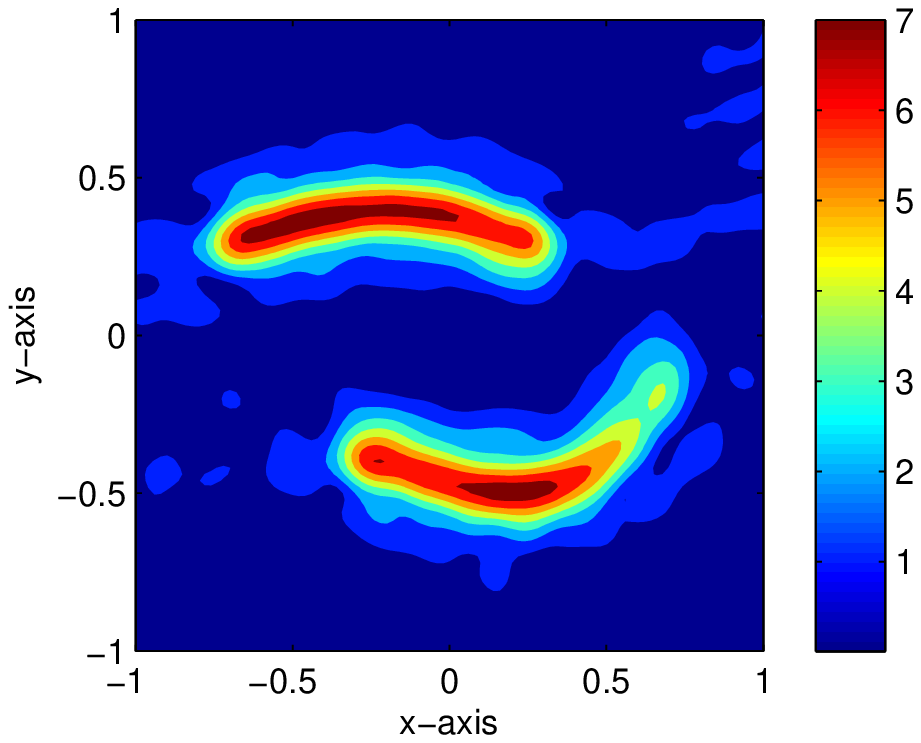}
\includegraphics[width=0.325\textwidth]{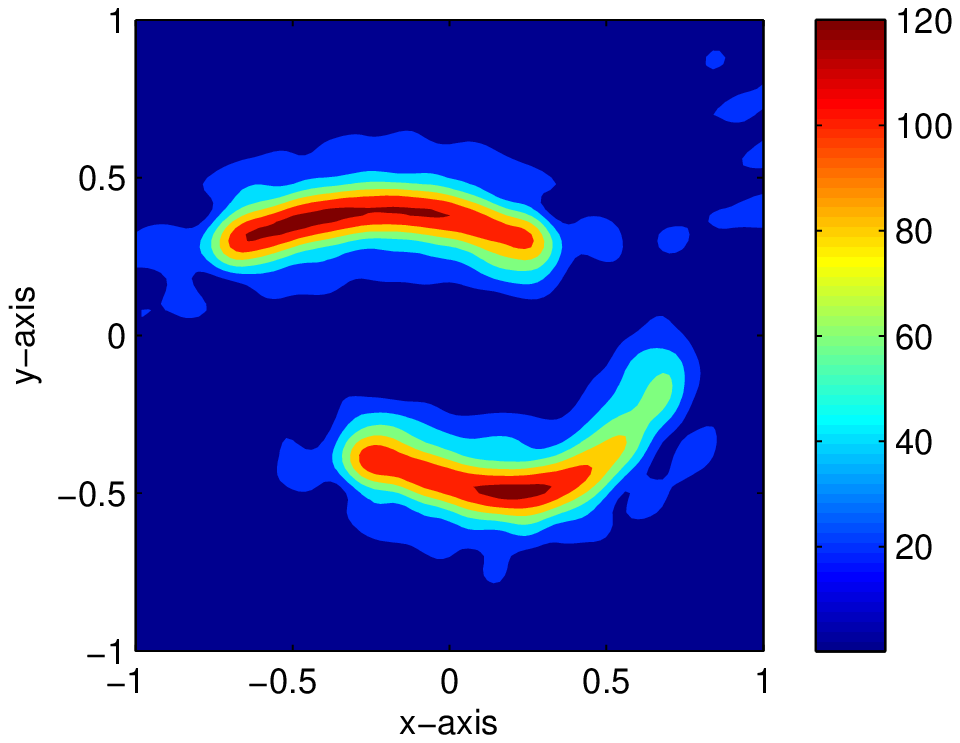}
\includegraphics[width=0.325\textwidth]{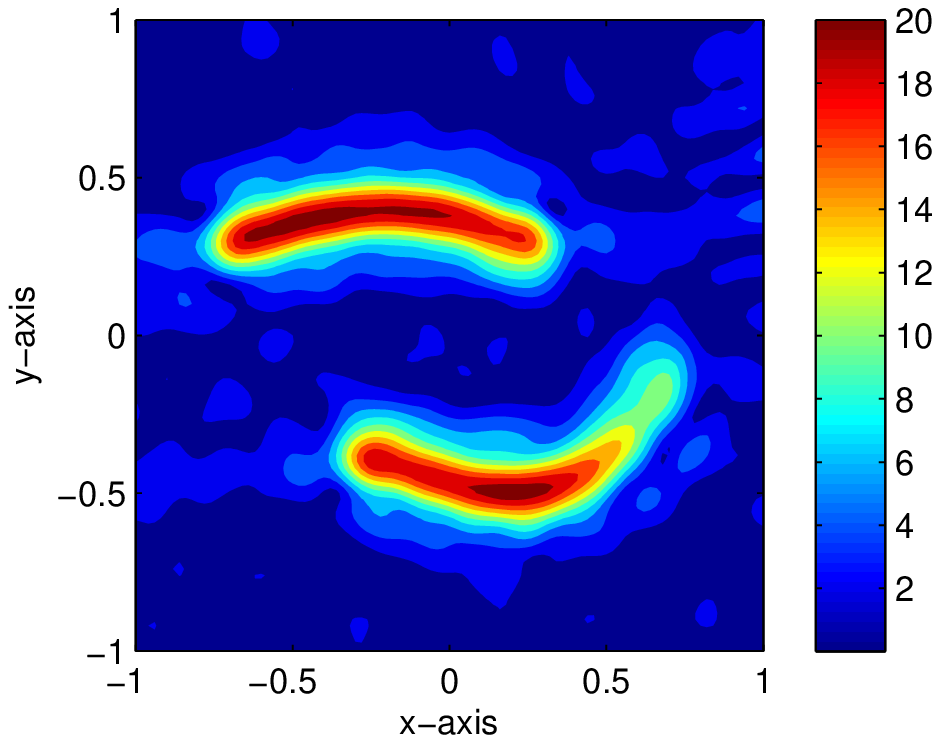}
\caption{\label{FigureM1}Maps of $\W_{\mathrm{MF}}(\mz;10)$ (left), $\W_{\mathrm{WMF}}(\mz;10,1)$ (center), and $\W(\mz;10)$ (right) when
the thin inclusion is $\Gamma_{\mathrm{M}}$ with same material property.}
\end{center}
\end{figure}

For the final example, we consider the imaging of multiple inclusions $\Gamma_{\mathrm{M}}$ with different properties $\eps_1=\mu=5$ and
$\eps_2=\mu_2=10$. Results in numerical simulations are exhibited in Figure \ref{FigureM2}. By comparing the results in Figure \ref{FigureM1}, we can observe that same as the imaging of single inclusion, almost every artifact has disappeared in the map of
$\W(\mz;10)$ while some of they are still remaining in the maps of $\W_{\mathrm{MF}}(\mz;10)$, $\W_{\mathrm{WMF}}(\mz;10,1)$.

\begin{figure}[!ht]
\begin{center}
\includegraphics[width=0.325\textwidth]{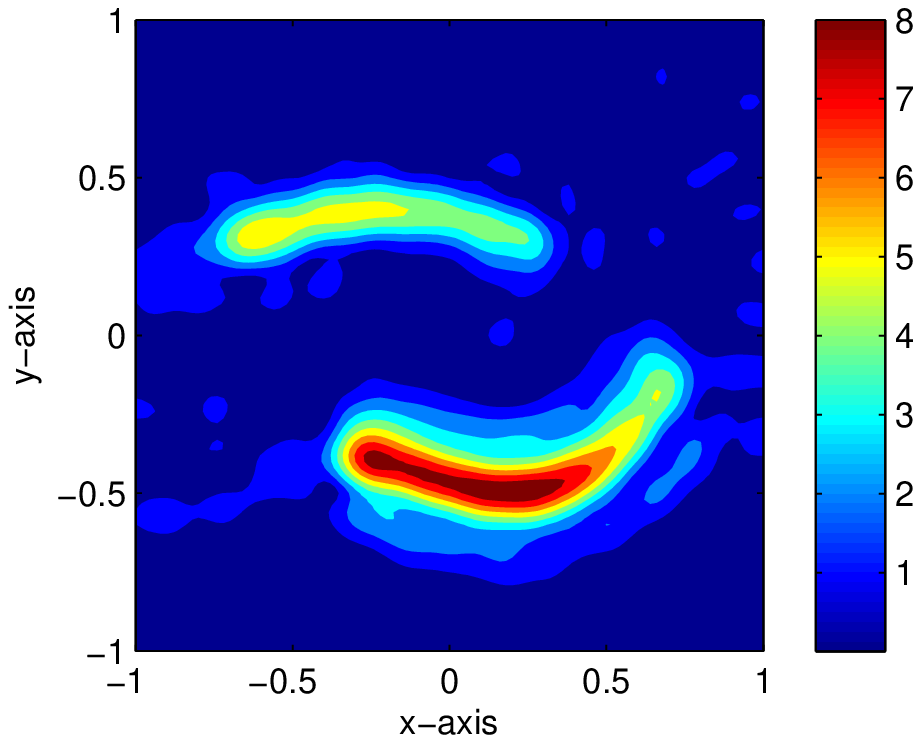}
\includegraphics[width=0.325\textwidth]{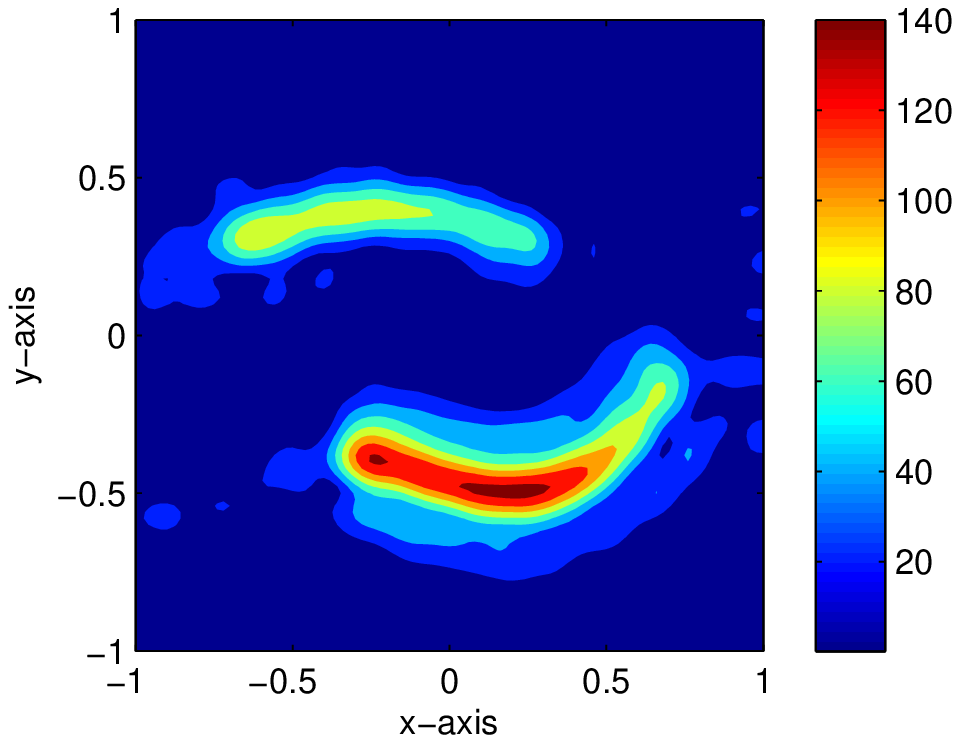}
\includegraphics[width=0.325\textwidth]{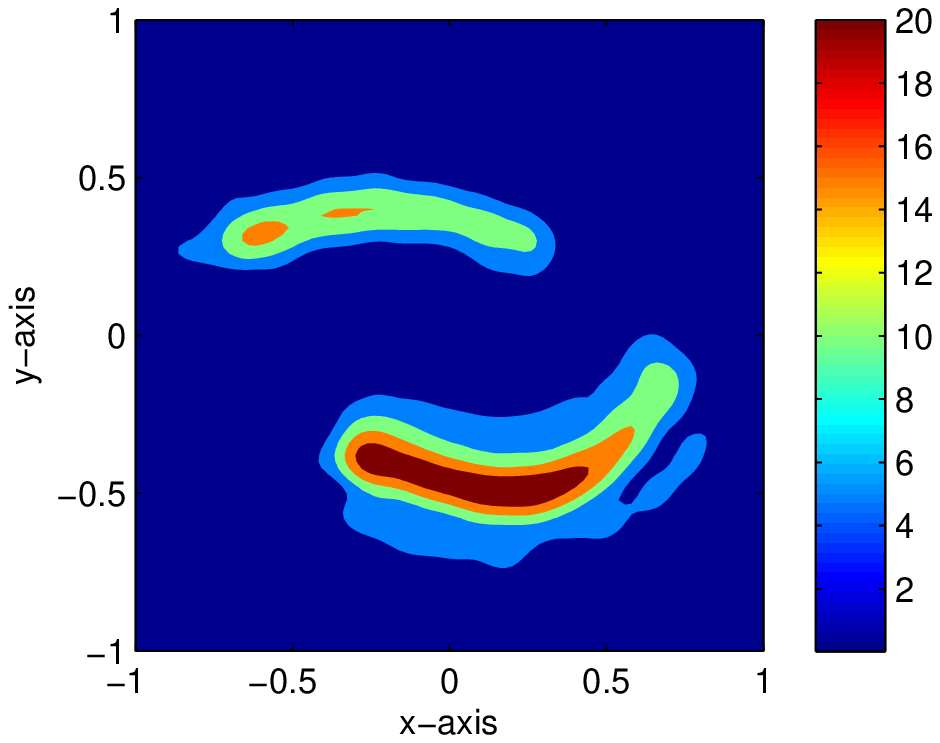}
\caption{\label{FigureM2}Same as Figure \ref{FigureM2} except different material properties.}
\end{center}
\end{figure}

To conclude this section, let us present some remarks. From the derivation of Theorem \ref{StructureWMSM}, it follows that the number $N$ of incident and observation directions and the number $F$ of applied frequencies have to be large enough. Furthermore, the selection of $\mc$ of (\ref{VecW}) must follow the unit normal direction on the supporting curve $\sigma$. On the other hand, we point out that if there are two inclusions with the same material property, our analysis and observation are not valid anymore. In fact, the derivation of asymptotic
expansion formula in the existence of multiple inclusions, one must assume that they are well-separated to each other. Hence, we expect that if the distance between two inclusions are sufficiently large, the result will be nice. But, this fact does not guarantee the improvement of proposed algorithm. Thus, finding a method of improvement is still required. Finally, we believe that although the result in this paper does not guarantee the complex shape of thin inclusions due to the intrinsic Rayleigh resolution limit, they can be good initial guesses of a
level-set method or of a Newton-type reconstruction algorithm, refer to \cite{ADIM,BHR,DL,HSZ,PL4,S} and references therein.

\section{Conclusion}\label{Sec5}
In the present study, we have proposed a multi-frequency subspace migration weighted by the natural logarithmic function for imaging of
thin, crack-like electromagnetic inclusions. This is based on the asymptotic expansion formula of far-field pattern in the existence of such inclusions and the structure of constructed MSR matrix operated at multiple frequencies. Throughout a careful analysis and numerical
experiment, it is confirmed that proposed method successfully improves traditional approaches. However, a counter example was discovered
when one tries to find the shape of multiple inclusions with the same material properties. Hence, investigating the reason will be an
interesting work. Furthermore, for achieving the best imaging of inclusions, finding {\it a priori} information of supporting curve, e.g.,
unit outward normal vector, should be a remarkable research. In this paper, we considered the imaging of inclusions located in the homogeneous space but based some recent works \cite{P3,PL2,PP}, subspace migration can be applicable for imaging of targets buried in the
half-space. Hence, an extension to the half-space problem is expected. And, similarly to \cite{KP}, the improvement considered herein can
be extended to the limited-view inverse scattering problems.

\end{document}